\documentclass[10xpt]{article}
\usepackage{amsfonts, amssymb, amsthm, amsmath}
\usepackage{xspace}
\input xy
\xyoption{all}
\renewcommand{\L}{{\mathcal L}}

\newcommand{\Tr}{\mbox{Tr}}
\newcommand{\tr}{\text{Tr}}

\newcommand{\hgamma}{\hat{\gamma}}
\newcommand{\hf}{\hat{f}}

\newcommand{\maxtensor}{\otimes_{max}}
\newcommand{\mintensor}{\otimes_{min}}
\newcommand{\tensor}{\otimes}
\renewcommand{\hat}{\widehat}

\newcommand{\Ordlin}{\mbox{{\bf Ordlin}}}

\newcommand{\FdHilb}{\mbox{{\bf FDHilb}}}

\newcommand{\Rvec}{\mbox{\bf RVec}}
\newcommand{\RVec}{\mbox{\bf Vec}_{\mathbb{R}}}

\newcommand{\Com}{{\bf Com}}
\newcommand{\FdCom}{{\bf FDCom}}

\newcommand{\B}{\mathfrak{B}}
\renewcommand{\L}{{\cal L}}

\renewcommand{\H}{{\bf H}}

\renewcommand{\1}{{\bf 1}}

\newcommand{\ket}[1]{| #1 \rangle}
\newcommand{\bra}[1]{\langle #1 |}

\newcommand{\homega}{\hat{\omega}}

\def\qed{$\Box$}

\newcommand{\id}[1]{\ensuremath{1_{#1}}}
\def\tr{\mbox{Tr}}

\newcommand{\op}{op}
\newcommand{\C}{{\mathcal C}}
\newcommand{\D}{{\mathcal D}}
\renewcommand{\text}{\mbox}
\renewcommand{\1}{{\sc 1}}

\newcommand{\beq}{\begin{equation}}
\newcommand{\eeq}{\end{equation}}
\newcommand{\beqa}{\begin{eqnarray}}
\newcommand{\eeqa}{\end{eqnarray}}
\def\cl{{\cal L}}

\def\R{{\mathbb R}}

\newcommand{\iso}{\cong}

\newcommand{\innp}[2]{%  dirac style inner product
    \ensuremath{\langle #1 \mid #2 \rangle}}

\newcommand{\fdhilb}{% the category of finite dimensional hilbert space
\ensuremath{\textbf{FDHilb}}\xspace}
\newcommand{\fdHilb}{\fdhilb}

\theoremstyle{plain}
\newtheorem{theorem}{Theorem}
\newtheorem{proposition}[theorem]{Proposition}
\newtheorem{lemma}[theorem]{Lemma}
\newtheorem{corollary}[theorem]{Corollary}

\theoremstyle{definition}
\newtheorem{definition}[theorem]{Definition}
\newtheorem{example}[theorem]{Example}
\newtheorem{examples}[theorem]{Examples}

\theoremstyle{remark}
\newtheorem{remark}[theorem]{Remark}

\newcommand{\tempout}[1]{{}}

\newcounter{thingy}
\renewcommand{\thethingy}{\arabic{section}.\arabic{thingy} }

\renewcommand{\section}[1]{ \refstepcounter{section}
\setcounter{thingy}{0} \vspace{.3in} \noindent{ \bf
\arabic{section}.~#1 }\vspace{.2in} \noindent}{}

\renewcommand{\subsection}[1]{ \refstepcounter{subsection}
\setcounter{thingy}{0}
\vspace{.2in} \noindent{ \bf \arabic{section}.\arabic{subsection}~#1 }}{}

\title{Symmetry, Compact Closure and Dagger Compactness for Categories of Convex Operational Models}
\author{Howard
  Barnum\footnote{Perimeter Institute for Theoretical Physics {\tt hbarnum@perimeterinstitute.ca}}, Ross
  Duncan\footnote{Oxford University Computing Laboratory, {\tt ross.duncan@comlab.ox.ac.uk}} and
  Alexander Wilce\footnote{Department of Mathematics, Susquehanna
    University, {\tt wilce@susqu.edu}}} \date{April 13, 2010}

\begin{document}

\maketitle

%\iffalse

\begin{abstract}
In the categorical approach to the foundations of quantum theory,
one begins with a symmetric monoidal category, the objects of which
represent physical systems, and the morphisms of which represent
physical processes.  Usually, this category is taken to be at least
compact closed, and more often, dagger compact, enforcing a certain
self-duality, whereby preparation processes (roughly, states) are
interconvertible with processes of registration (roughly,
measurement outcomes). This is in contrast to the more concrete
``operational" approach, in which the states and measurement
outcomes associated with a physical system are represented in terms
of what we here call a {\em convex
  operational model}: a certain dual pair of ordered linear spaces --
generally, {\em not} isomorphic to one another. On the other hand,
state spaces for which there {\em is} such an isomorphism, which we
term {\em weakly self-dual}, play an important role in
reconstructions of various quantum-information theoretic protocols,
including teleportation and ensemble steering.

In this paper, we characterize compact closure of symmetric monoidal
categories of convex operational models in two ways: as a statement
about the existence of teleportation protocols, and as the principle
that every process allowed by that theory can be realized as an
instance of a remote evaluation protocol --- hence, as a form of
classical probabilistic conditioning. In a large class of cases,
which includes both the classical and quantum cases, the relevant
compact closed categories are degenerate, in the weak sense that
every object is its own dual. We characterize the dagger-compactness
of such a category (with respect to the natural adjoint) in terms of
the existence, for each system, of a {\em symmetric} bipartite
state, the associated conditioning map of which is an isomorphism.
\end{abstract}

\section{Categorical Semantics and Quantum Foundations}
\label{sec:intro}

One natural way to formalize a physical theory is as some kind of {\em
  category}, $\C$, the objects of which are the systems, and the
morphisms of which are the {\em processes}, contemplated by that
theory. In order to provide some apparatus for representing compound
systems, it is natural to assume further that $\C$ is a {\em symmetric
  monoidal} category. In the categorical semantics for quantum theory
pioneered by Abramsky and Coecke \cite{Abramsky-Coecke}, Selinger
\cite{Selinger, SelingerCPM}, and others (e.g.,
\cite{ACHbk,AD,Baez04a,CPavPaq}), it is  further assumed that $\C$
is at least compact closed, and more usually, dagger compact. This
last condition enforces a certain self-duality, in that there is a
bijection between the states of a system $A \in \C$,
represented by elements of $\C(I,A)$, and and the
measurement-outcomes associated with that system, represented by
elements of $\C(A,I)$. The motivating example here is the category
$\FdHilb$ of finite-dimensional complex Hilbert spaces and unitary
mappings --- that is, the category of finite-dimensional ``closed"
quantum systems and unitary processes. Many of the
information-processing features of finite-dimensional quantum
systems occur in any dagger-compact category, notably, conclusive
(that is, post-selected) teleportation and entanglement-swapping
protocols.  On the other hand, if our interest in a categorical reformulation of
quantum theory is mainly foundational, rather than strictly one of
systematization, these strong structural assumptions need further
justification, or at any rate, further motivation.

There is an older tradition, stemming from Mackey's work on the
foundations of quantum mechanics \cite{Mackey}, in which an individual
physical (or, more generally, probabilistic) system is represented by
a set of states, a set of observables or measurements, and an
assignment of probabilities to measurement outcomes, conditional upon
the state. From this basic idea, one is led to a representation of
systems involving pairs of ordered real vector spaces --- the {\em
  convex operational models} of our title --- and of physical processes, by certain positive linear mappings
between such spaces.  The motivating example is the category
of Hermitian parts of $C^{\ast}$ algebras and completely positive
mappings.

This ``convex operational" approach, in contrast to the categorical
one, is conservative of classical probabilistic concepts, but
liberal as to how systems may be combined and transformed, so long
as this probabilistic content is respected. In particular, there is
no standing assumption of monoidality; rather, systems are combined
using any of a variety of ``non-signaling" products. Nor is there,
in general, any hint of the kind of self-duality mentioned above ---
indeed, the natural dual object for a convex operational model is
not itself an operational model.
 Nevertheless, here again various familiar ``quantum"
phenomena -- such as no-cloning and no-broadcasting theorems,
information-disturbance tradeoffs, teleportation and entanglement
swapping protocols, and ensemble steering -- emerge naturally and in
some generality \cite{Barrett, BBLW06, BBLW08, BGW09}.  A key idea here is that
of a {\em remote evaluation protocol} \cite{BBLW08} (of which
teleportation is a special case), which reduces certain kinds of
dynamical processes to purely classical {\em conditioning}.

It is obviously of interest to see how far such convex operational
theories can be treated formally, that is, as categories, and more
especially, as symmetric monoidal categories; equally, one would like
to know how much of the special structure assumed in the categorical
approach can be given an operational motivation.\footnote{That {\em
    some} motivation is needed is clear in view of a result, to
    appear elsewhere, that any of a broad class of symmetric monoidal
  categories can be interpreted as categories of convex operational
  models.} Some first steps toward addressing these issues are taken
in \cite{BW09a, BW09b}. Here, we aim to make further progress, albeit
along a somewhat narrower front. We focus on {\em symmetric monoidal}
categories of convex operational models -- what we propose to call
{\em probabilistic theories}. We show that such a theory admits a
compact closed structure if and only if every system allowed by the
theory can be teleported (conclusively, though not necessarily with
probability $1$) through a copy of itself -- or, equivalently, if and
only if every process contemplated by the theory can be represented as
a remote evaluation protocol involving a copy of itself.  We then
specialize further, to consider {\em weakly self-dual} theories, in
which for every system $A$ there is a bipartite state $\gamma_A$ on $A
\otimes A$ corresponding to an isomorphism between $A$ and its dual,
and an effect corresponding to its inverse. (Such state spaces figure
heavily in earlier treatments of teleportation protocols \cite{BBLW08}
and ensemble steering \cite{BGW09} in general probabilistic theories.)
We show that if the state implementing weak self-duality can be chosen
to be symmetric for every $A$, then a weakly self-dual monoidal
probabilistic theory is not merely compact closed, but
dagger compact. \\

\tempout{ We also observe that any symmetric monoidal category can
be represented functorially in a category of convex operational
models, given only an interpretation of the ``scalars" -- the
endomorphisms of the tensor unit -- as determining genuine
probabilities, i.e, as values in $[0,1]$. \\}

\noindent{\bf Organization and Notation} Sections
\ref{sec:category-perspective} and \ref{sec:abstr-state-spac}
provide quick reviews of the category-theoretic and the convex
frameworks, respectively, mainly following \cite{Abramsky-Coecke}
for the former and \cite{BBLW06, BBLW08, BGW09, BW09a, BW09b} for
the latter. Section \ref{sec:mono-prob-theor} makes precise what we
mean by a monoidal probabilistic theory, as a symmetric monoidal
category of convex operational models, and establishes that all such
theories have the property of allowing {\em remote evaluation}
\cite{BBLW08}; when the state spaces involved are weakly self-dual,
teleportation arises as a special case. Section
\ref{sec:weakly-self-dual} contains the results on categories of
weakly self-dual state spaces described above.
Section \ref{sec:conclusion} discusses some of the further
ramifications of these results.

We assume that the reader is familiar with basic category-theoretic
ideas and notation, as well as with the probabilistic machinery of
quantum theory. We write $\C, \D$ etc. for categories, $A \in \C$, to
indicate that $A$ is an object of $\C$, and $\C(A,B)$ for the set of
morphisms between objects $A, B \in \C$.  Except as noted, all vector
spaces considered here will be finite-dimensional and {\em real}. We
write $\RVec$ for the category of finite-dimensional real vector
spaces and linear maps. The dual space of a vector space $A$ is
denoted by $A^{\ast}$.  An {\em ordered vector
  space} is a real vector space $V$ equipped with a \emph{regular} ---
that is, closed, convex, pointed, generating --- cone $V_{+}$, and
ordered by the relation $x \leq y \Leftrightarrow y - x \in A_{+}$.
A linear mapping $\phi : V \rightarrow W$ between ordered linear
spaces $V$ and $W$ is {\em positive} if $\phi(V_{+}) \subseteq
W_{+}$. We write $\L_{+}(A,B)$ for the cone of positive linear
mappings from $A$ to $B$.  The special case in which $B = \R$, the
positive linear functionals on $A$, is the \emph{dual cone} of
$A_+$, denoted $A_+^*$.  The category of ordered linear spaces and
positive linear maps we denote by $\Ordlin$. Finally, we make 
the standing assumption that,
except where otherwise indicated, {\em all
vector spaces considered here are finite dimensional}.\\

\tempout{Note that our definition of ordered linear space differs
from the usual one in requiring the cone to be closed and
generating.  In the remainder of the paper we use the term ``cone''
to mean ``regular cone'' unless otherwise noted.  \\}

\noindent{\bf Acknowledgements} HB and AW wish to thank Samson
Abramsky and Bob Coecke for enabling them to visit the Oxford
University Computing Laboratory in November 2009, where some of this
work was done, and for helpful discussions during that time.  The
authors also wish to thank Peter Selinger for helpful discussions.  RD
is supported by EPSRC postdoctoral research fellowship EP/E04006/1.
HB's research was supported by Perimeter Institute for Theoretical Physics;
work at Perimeter Institute is supported in part by the Government of
Canada through Industry Canada and by the Province of Ontario through
the Ministry of Research and Innovation.

\section{The Category-Theoretic Perspective}
\label{sec:category-perspective}

A {\em monoidal category} \cite{MacLane} is a
category $\C$ equipped with a bifunctor\footnote{
  Bifunctoriality means that: (i) $1_{A\otimes B} = 1_A \otimes
  1_B$; and (ii) given morphisms $f : A \rightarrow X$ and $g : B
  \rightarrow Y$ in $\C$, there is a canonical product morphism $f
  \otimes g : A \otimes B \rightarrow X \otimes Y$, such that
    $    (f \otimes g) \circ (f' \otimes g') =
    (f \circ f') \otimes (g \circ g')$.
} $\otimes : \C \times \C
\rightarrow \C$, a distinguished unit object $I$, and natural
associativity and left and
right unit deletion isomorphisms,
\[
\begin{split}
  \alpha_{A,B,C}: A \otimes (B \otimes C) \iso (A \otimes B) \otimes C,
  \\
\lambda_A : I \otimes A \iso A, \qquad\qquad \rho_A : A \otimes I \iso A  ,
\end{split}
\]
subject to some coherence conditions (for which, see \cite{MacLane}).  If these isomorphisms are
identities, the category is called \emph{strict monoidal}.  Every
monoidal category  is equivalent to a strict one, so setting $A\otimes I
= A$ etc, is harmless and we will do this throughout.

A \emph{symmetric monoidal  category} (SMC) is a
monoidal category further equipped with a natural family of
symmetry isomorphisms,
\[
\sigma_{A,B} :A \otimes B \iso B \otimes A,
\]
again, subject to some coherence conditions (for which, again, see \cite{MacLane}).  Unlike the other isomorphisms,
these symmetry isomorphisms cannot generally be made strict.

Examples of SMCs include commutative monoids (as one-object categories),
the category of sets and mappings (with $A \otimes B = A \times B$),
and -- of particular relevance for us -- the category of (say,
finite-dimensional) vector spaces over a field $K$ and $K$-linear
maps, with $A \otimes B$ the usual tensor product.
Another source of examples comes from logic: one can regard the set of sentences of
a logical calculus as a category, with {\em proofs}, composed by concatenation, as morphisms. In
this context, one can take conjunction, $\wedge$, as a monoidal
product.

Much more broadly, if somewhat less precisely, if one views the
objects of a category $\C$ as ``systems" (of whatever sort), and
morphisms as ``processes" between systems, then a natural
interpretation of the product in a symmetric monoidal category
is as a kind of accretive composition: $A
\otimes B$ is the system that consists of the two systems $A$ and $B$
sitting, as it were, side by side, without any special interaction; $f
\otimes g$ represents the processes $f : A \rightarrow X$ and $g : A
\rightarrow Y$ acting {\em in parallel}. Taking this point of view, it is helpful to regard processes of the form $I \rightarrow A$,
where $I$ is the monoidal unit in $\C$, as
{\em states}, associated with ways of preparing the system $A$.
%(Thus, concretely, $\C(I,A)$ is the {\em state space} of the system $A$.)Let $X$ and $\Sigma $ be
Similarly, we regard processes of the form $a : A \rightarrow I$ as
``effects'', or measurement-outcomes.  We shall henceforth adhere to
the convention of denoting states by lower-case Greek letters $\alpha,
\beta, ...$ and effects, lower-case Roman letters $a, b, ...$.

In any monoidal category, one can regard endomorphisms $s \in
\C(I,I)$ as ``scalars'' acting on elements of $\C(A,B)$ by $sx = s
\otimes x$.  In every monoidal category, $\C(I,I)$ is a commutative
monoid, even if $\C$ is not symmetric.  When
$\C(I,I)$ is isomorphic to a particular monoid $S$, we shall say that
$\C$ is a symmetric monoidal category {\em over} $S$.

%\newpage
\subsection{Compact Closed Categories}

\label{sec:comp-clos-categ} A {\em dual} for an object $A$ of a
symmetric monoidal category $\C$ is an object $B$ and two morphisms,
the \emph{unit}, $\eta : I \rightarrow B \otimes A$ (not to be
confused with the \emph{tensor unit} $I$) and the \emph{co-unit},
$\epsilon : A \otimes B \rightarrow I$, such that
\begin{equation}
\begin{array}{ccc}
{\xymatrix@=12pt{
 A  \ar@{->}^{\id{A} \otimes \eta\quad}[rrr]  & & & A \otimes B \otimes A
\ar@{->}^{\epsilon \otimes \id{A}}[rrr] &&& A}} = \id{A}
\\
{\xymatrix@=12pt{
 B  \ar@{->}^{\eta \otimes \id{B}\quad}[rrr]  & & & B \otimes A \otimes B
\ar@{->}^{\id{B} \otimes \epsilon}[rrr] &&& B}} = \id{B}.
\end{array}
\end{equation}

Duals are unique up to a canonical isomorphism. Indeed, if
$(B_1,\eta_1,\epsilon_1)$ and $(B_2,\eta_2,\epsilon_2)$ are duals for
$A$, then $\phi := (\id{B_2} \otimes \epsilon_1) \circ (\eta_{2}
\otimes \id{B_1}) : B_1 \rightarrow B_2$ has inverse $\phi^{-1} = (
\id{B_1}\otimes\epsilon_2) \circ (\eta_{1} \otimes \id{B_2})$;
moreover, $\eta_2 = (\phi \otimes \id{A}) \circ \eta_1$.  Some, for
example the authors of \cite{CPavPaq}, use the term \emph{compact
  structure} to refer to what we are calling a dual, i.e., a
particular choice of $(A',\eta_A,\epsilon_A)$ for a given object $A
\in \C$ or, when applied to a category, a particular choice of
dual for each object.

A symmetric monoidal category $\C$ is \emph{compact
  closed}\footnote{Sometimes just \emph{compact}.}  if for every
object $A$ in the category, there is a dual, $(A',\eta_A,\epsilon_A)$,
where $A'$ is an object of the category.\footnote{We use the notation
  $A'$, rather than the more standard $A^{\ast}$, for the designated
  dual of an object in a compact closed category, because we wish to
  reserve the latter to denote, specifically, the dual space of a
  vector space.}  As thus defined, compact closedness is a
\emph{property} of the SMC $\C$, not an additional structure: it
requires the existence of at least one dual for each object, but not
the explicit specification of a distinguished one.  The alternative
definition of compact closed category, which differs only in requiring
a choice of duals be specified \cite{Kelly-Laplaza}, is perhaps more
common.  Owing to the uniqueness up to isomorphism mentioned above,
the various possible choices of duals are largely---but not
entirely---equivalent. A compact structure is said to be {\em degenerate} iff $A' = A$; a
compact closed category $\C$ \emph{with a distinguished compact structure} is
said to be degenerate if every object's compact structure is
degenerate.  This \emph{does} depend on an explicit choice of duals,
and thus imposes some non-trivial structure beyond compact closure.
This is the setting that will most interest us below.

\begin{remark}
  If $(A',\eta_A,\epsilon_A)$ and $(B',\eta_B, \epsilon_B)$ are duals
  for objects $A, B \in \C$, then we can construct a canonical dual
  $(A' \otimes B', \eta_{AB}, \epsilon_{AB})$ for $A \otimes B$ by
  setting $\eta_{AB} = \tau \circ (\eta_A \otimes \eta_B)$ and
  $\epsilon_{AB} = (\eta_A \otimes \eta_B) \circ \tau^{-1}$, where
  \[\tau = 1_{A'} \otimes \sigma_{AB'} \otimes 1_B : (A' \otimes A) \otimes (B'
  \otimes B) \simeq (A' \otimes B') \otimes (A \otimes B).\] Since all
  duals are isomorphic we are free to assume that $A\otimes B$ has
  this particular dual. \footnote{Technically, this depends on
    the ability to factor objects uniquely as tensor products;
    however, all categories ordinarily considered in this context have
    this property.  The fact that the symmetry isomorphism can't
    generally be made strict is relevant here.}
\end{remark}

\tempout{As an example, let $\Rvec$ be the symmetric monoidal category of finite-dimensional real vector spaces and linear mappings, with the usual tensor product. In this case, if $(B,\eta,\epsilon)$ is a dual for $A \in \Rvec$,
then $\epsilon : A \otimes B \rightarrow I$ amounts to a non-degenerate bilinear form on $A \times B$, whereby
we can identify $B$ with $A^{\ast}$. } %%Ad example of FD Hilb?

In any compact closed category, an assignment $A \mapsto A'$
extends to a contravariant functor $(-)' :
\C^{\op} \rightarrow \C$, called the {\em adjoint}, taking morphisms
$\phi : A \rightarrow B$  to $\phi' : B' \rightarrow A'$  defined by:
\begin{equation}
{\xymatrix@=12pt{ B' \ar@{->}^{\eta_A \otimes \id{B'}}[rrr]
\ar@{->}^{\phi'}[dd] & & \hspace{.2in} & A' \otimes A \otimes B'
\ar@{->}^{\id{A'} \otimes \phi  \otimes \id{B'}}[dd] \\
& & \hspace{.2in} & \\
A' & & \hspace{.2in} & A'\otimes B \otimes B'. \ar@{->}^{\id{A'} \otimes
\epsilon_{B}}[lll]}}
\end{equation}

The functor $'$ is {\em nearly} involutive, in that there are natural isomorphisms
$w_A: A'' \rightarrow A$. To say that $'$ is involutive is just to say
that $A''=A$ and $\phi = \phi''$; note that this does not imply
that $w_A = \id{A}$.

\begin{remark} In the classic treatment of coherence for compact closed categories in
\cite{Kelly-Laplaza}, one has that $ \sigma \circ \eta_{A} = (1_A
\otimes w_A) \circ \eta_{A'}$; a similar condition holds for
$\epsilon$ (\cite{Kelly-Laplaza}, eq. (6.4)ff.).  In the case of a
degenerate category, this implies that 
\beq \label{eq:  near symmetry of unit} 
    \sigma \circ \eta_{A} = (1_A \otimes w_A) \circ \eta_{A}.  
\eeq
It is easy to show that if the units -- or, equivalently, co-units -- are
symmetric, in the sense that, for every object $A \in \C$, $\eta_A =
\sigma_{A,A} \circ \eta_{A}$ or, equivalently, $\epsilon_{A} =
\epsilon_{A} \circ \sigma_{A,A}$, then the the functor $'$ is
involutive.   However the converse does not necessarily hold, unless
$w_A = \id{A}$.  In general, it is not clear what coherence requirements 
are appropriate for the degenerate categories we consider, nor whether the 
functors involved are always strict.  Therefore, in Section
\ref{sec:weakly-self-dual} we will establish explicitly that the
involutiveness of the adjoint is equivalent to the symmetry of the
unit and co-unit for the compact closed categories of convex
operational models considered in this paper. \end{remark}

\begin{remark} For an arbitrary degenerate compact closed category,
there is no guarantee that the unit $\eta_A$ will be symmetric. (We thank
Peter Selinger \cite{SelingerExample} for supplying a nice example
involving a category of plane tangles.)  Thus, it is a non-trivial
constraint on such a category that the canonical adjoint be an
involution. This will be important below.\end{remark}

\subsection{Daggers}

\label{sec:daggers} A {\em dagger category} \cite{Selinger,
Abramsky-Coecke} is a category $\C$ together with an involutive
functor $(-)^{\dagger} : \C^{\op} \rightarrow \C$ that acts as the
identity on objects.  That is, $A^{\dagger} = A$ for all $A
\in \C$, and, if $f \in \C(A,B)$, then $f^{\dagger} \in \C(B,A)$,
with
\[
(g \circ f)^{\dagger} = f^{\dagger} \circ g^{\dagger} \quad \text{and} \quad
f^{\dagger \dagger} = f
\]
for all $f \in \C(A,B)$ and $g \in \C(B,C)$.\footnote{ Those new to
  categories should note that a functor from $\C^{\op}$ to $\C$ is
  sometimes called a \emph{contravariant} functor from $\C \rightarrow
  \C$; the description we have just given (minus the involutiveness
  condition) defines this notion without reference to $\C^{\op}$.}  We
say that $f$ is {\em unitary} iff $f^{\dagger} = f^{-1}$.  A {\em
  dagger-monoidal} category is a symmetric monoidal category with a
dagger such that (i) all the canonical isomorphisms defining the
symmetric monoidal structure are unitary, and (ii)
\[
(f \otimes g)^{\dagger} = f^{\dagger} \otimes g^{\dagger}
\]
for all morphisms $f$ and $g$ in $\C$.
Finally, a dagger-monoidal category $\C$ is {\em dagger compact}
if  it is compact closed and
\[
\eta_{A}  =  \sigma_{A,A'} \circ \epsilon_{A}^{\dagger}
\]
for every $A$, i.e.:
\[{\xymatrix@=12pt{
& & A \otimes A' \ar@{->}^{\sigma}[dd]\\
I \ar@{->}^{\epsilon_{A}^{\dagger}}[rru] \ar@{->}^{\eta_{A}}[rrd] & & \\
& & A' \otimes A }}\]
commutes. In the case of a degenerate compact closed category, the canonical
adjoint $'$ functions as a dagger {\em if} it is involutive. However, as remarked
above, this is a nontrivial condition.

In the work of Abramsky and Coecke \cite{Abramsky-Coecke}, a
symmetric monoidal category is interpreted as a physical theory, in
which a morphism $\alpha : I \rightarrow A$ is interpreted, as discussed above, as 
representing a {\em state} of the system $A$; a morphism $b = \beta^\dag : A \rightarrow I$ is
understood as the registration of an effect --- e.g., a measurement
outcome --- associated with $A$. The scalar $\beta^\dag \circ \alpha : I \rightarrow I$ is
understood, somewhat figuratively in the abstract setting, as the
``probability" that the given effect will occur when the given state
obtains.%
\footnote{We can be more precise here: given a dagger-compact category  
of states and processes $\C$, any dagger-monoidal functor from $\C$
 to \fdHilb, the category of finite dimensional Hilbert spaces, will
 send the scalar $\beta^\dag \circ \alpha$ to the inner product
 $\innp{\beta}{\alpha}$.}
This raises
the obvious question of how to implement the compelling idea
that probabilities should be identified with real numbers in the
interval $[0,1]$ without passing through Hilbert space.
One way to do this is simply to posit a
mapping $p : \C(I,I) \rightarrow [0,1]$, whereby the scalars
of $\C$ can be interpreted probabilistically. Another is to
examine the symmetric monoidal possibilities in cases in which the
category consists, {\em ab initio}, of concretely described
probabilistic models of a reasonably simple and general sort. In this
paper, we concentrate on this second strategy.
As a first step, in the next section we describe the kinds of concrete probabilistic models
we have in mind.

\section{Convex Operational Models and their Duals}
\label{sec:abstr-state-spac}

The more traditional approach to modeling probabilistic physical
theories \cite{Mackey, Beltrametti-Cassinelli} begins by associating
to each individual physical (or other probabilistic) system a triple
$(X,\Sigma,p)$ -- sometimes called a {\em Mackey triple}
--- where $\Sigma$ is a set of possible {\em states}, $X$ is a set
of possible measurement-outcomes, and $p : X \times \Sigma
\rightarrow [0,1]$ assigns to each pair $(x,s)$ the probability,
$p(x,s)$, that $x$ will occur, if measured, when the system's state
is $s$.

This minimal apparatus can be ``linearized'' in a natural way.  
The probability function $p$ gives us a mapping $\Sigma \rightarrow [0,1]^{X}$, 
namely $s \mapsto p( \cdot , s)$. We can plausibly identify each state $s \in \Sigma$ 
with its image under this mapping (thus identifying states if they cannot be
distinguished statistically by the outcomes in $X$). Having done so,
let $\Omega$ denote the point-wise closed, and hence compact, convex
hull of $\Sigma \subseteq [0,1]^{X}$. This represents the set of
possible probabilistic {\em mixtures} of states in $\Sigma$, in so 
far as these can be distinguished by outcomes in $X$. Every
measurement outcome $x \in X$ can now be represented by the affine
evaluation functional $a_x : \Omega \rightarrow [0,1]$, given by $a_x(\alpha) =
\alpha(x)$ for all $\alpha \in \Omega$. More broadly, we can regard
{\em any} affine functional $a : \Omega \rightarrow [0,1]$ as
representing a mathematically possible measurement outcome, having
probability $a(\alpha)$ in state $\alpha \in \Omega$. Such
functionals are called {\em effects} in the literature.

In general, the (mixed) state space $\Omega$ that we have just
constructed will have infinite affine dimension. Accordingly, for
the next few paragraphs, we suspend our standing
finite-dimensionality assumption. Now, any compact convex set
$\Omega$ can be embedded, in a canonical way, as a {\em base} for
the positive cone $V_{+}(\Omega)$ of a regularly ordered linear
$V(\Omega)$ \cite{Alfsen}. This means that every $\rho \in
V_{+}(\Omega)$ has the form $\rho = t \alpha$ for a unique scalar $t
\geq 0$ and a unique vector $\alpha \in \Omega$ (hence, $\Omega$
spans $V(\Omega)$). This space $V(\Omega)$ is complete in a natural
norm, the {\em base norm}, the unit ball of which is given by the
closed convex hull of $\Omega \cup -\Omega$. Moreover, $V(\Omega)$
has the following universal property: every bounded affine mapping
$L : \Omega \rightarrow {\bf M}$, where ${\bf M}$ is any real Banach
space, extends uniquely to a bounded linear mapping $L : V(\Omega)
\rightarrow {\bf M}$.

In particular, every affine functional on $\Omega$ -- in particular,
every effect -- extends uniquely to a linear functional in
$V(\Omega)^{\ast}$. In particular, there is a unique {\em unit
functional} $u_\Omega \in V(\Omega)^{\ast}$ such that
$u_\Omega(\alpha) = 1$ for $\alpha \in \Omega$, and $\Omega =
u_\Omega^{-1}(1) \cap V_{+}(\Omega)$.  Thus, effects correspond to
positive functionals $a \in V(\Omega)^{\ast}$ with $0 \leq a \leq
u_\Omega$.

One often regards {\em any} effect $a \in V(\Omega)^{\ast}$ as a
{\em bona fide} measurement outcome. This is the point of view,
e.g., of \cite{BBLW06, BBLW08}. However, we may sometimes wish to
privilege certain effects as ``physically accessible". This suggests
the following more general formulation:

\begin{definition}\label{defn:COM}
A {\em convex operational model} (COM) is a triple $(A,A^{\#},u_A)$
where \begin{itemize} \item[(i)] $A$ is a complete base-normed space
with (strictly positive) unit functional $u_A$, and
\item[(ii)] $A^{\#}$ is a weak-$\ast$ dense subspace of $A^{\ast}$,
ordered by a chosen regular cone $A^{\#}_{+} \subseteq A^{\ast}_{+}$
containing $u_A$. \end{itemize} An {\em effect} on $A$ is a
functional $a \in A^{\#}_{+}$ with $a \leq u_A$.
\end{definition}

\noindent Henceforth, where no ambiguity seems likely, we write $A$
for the triple $(A,A^{\#},u_{A})$. Also, we now revert, for the
balance of this paper, to our standing assumption  that {\em all
COMs are finite-dimensional.}  In this case, the weak-$\ast$ density
assumption above simply says that $A^{\#} = A^{\ast}$, so that
$A^{\#} = A^{\ast}$ as vector spaces. Even in this situation,
however, the chosen cone $A^{\#}_{+}$ will generally be smaller than
the dual cone $A^{\ast}_{+}$, so the positive cone
$(A^{\#})_{+}^{\ast}$ will in general be {\em} larger than $A_{+}$.
It is useful to regard normalized elements of the former cone as
{\em mathematically consistent} probability assignments on the
effects in $A^{\#}$, from which the model singles out those in
$A_{+}$ as {\em physically} possible. In the special case in which
$A^{\#}_+ = A^{\ast}_+$ --- as, e.g., in the case of quantum systems
--- we shall say that the COM $A$ is {\em saturated}.

\begin{example}\label{ex:prob-simplex}
  Let $E$ be a finite set, thought of as the outcome set for a
  discrete classical experiment. Take $A = {\mathbb{R}}^{E}$, with
  $A_{+}$ the cone of non-negative functions on $E$, and let $u_{A}(f)
  = \sum_{x \in E} f(x)$. Then $\Omega = u^{-1}(1)$ is simply the set
  of probability weights on $E$. Geometrically, this last is a {\em
    simplex}.  In finite dimensions, every simplex has this
  form. Accordingly, we say a COM is {\em classical} iff  its normalized
  state space is a simplex.
\end{example}

\begin{example}\label{ex:densityoperators}
  Let $\H$ be a finite-dimensional complex Hilbert space, and let $A =
  \L_{h}(\H)$, the space of Hermitian operators $a : \H \rightarrow
  \H$, with the usual positive cone, i.e, $A_{+}$ consists of all
  Hermitian operators of the form $a^{\dagger} a$. Let $u_{A}(a) =
  \tr(a)$. Then $\Omega_{A}$ is the convex set of density operators on
  $\H$, i.e., the usual space of mixed quantum states.
\end{example}

\begin{example}\label{ex:mackey-triple}
  Let $(X,\Sigma,p)$ be any Mackey triple. Construct the state-space $\Omega$
  and the associated ordered Banach space $V(\Omega)$ as described above. Letting
  $A^{\#}_{+}$ be the cone in $V^{\ast}(\Omega)$ generated by the evaluation functionals
  $a_{x}$, $x \in X$, we have a convex operational model. This will
  be finite-dimensional iff the span of (the image of) $\Sigma$ in ${\Bbb R}^{X}$
  is finite dimensional.

\end{example}

\subsection{Processes as Positive Mappings}

\label{sec:proc-as-posit}
\begin{definition}\label{defn:COMmorphism}
A {\em morphism of COMs} from $(A,A^{\#},u_A)$ to $(B,B^{\#},u_B)$ is
a positive linear map $\phi : A \rightarrow B$ such that the usual linear
adjoint map $\phi^* : B^{*} \rightarrow A^{*}$ is positive with respect
to the designated cones $A^{\#}_{+}$ and $B^{\#}_{+}$.
\end{definition}

\noindent
The set of morphisms of COMs from $A$ to $B$ is clearly a sub-cone of
$\L_{+}(A,B)$. It is clear that the composition of two mappings of
COMs is again a mapping of COMs, so that COMs form a concrete category.

\begin{definition}\label{defn:process}
 Let $A$ and $B$ be COMs. A {\em process} from $A$ to $B$ is a morphism $\phi : A \rightarrow B$ such that, for
 every state $\alpha \in \Omega_A$, $u_{B}(\phi(\alpha)) \leq 1$, or,
 equivalently, if $\phi^{\ast}(u_{B}) \leq u_{A}$.
\end{definition}

If $\phi : A \rightarrow B$ is a process, we can
regard $u_{B}(\phi(\alpha))$ as the {\em probability} that the
process represented by $\phi$ occurs. If we regard ${\mathbb{R}}$ as an
COM with $u_{\mathbb{R}}$ the identity mapping on ${\mathbb{R}}$, this is
consistent with our understanding of $a(\alpha)$ as the probability
of the effect $a : A \rightarrow {\mathbb{R}}$ occurring. Notice that a
positive linear map $\phi : {\mathbb{R}} \rightarrow A$ is a process if
and only if $\phi(1)$ is a sub-normalized state, while a positive functional $f
: A \rightarrow {\mathbb{R}}$ is a process if and only if $f \in A^{\#}_{+}$ and $f
\leq u_A$ -- in other words, if and only if $f$ is an effect.
Finally, since $\Omega_{A}$ is compact,
$u_{A}(\phi(\alpha))$ attains a maximum value, say $M$ on
$\Omega_{A}$.  $M^{-1}\phi$ is a process, so every morphism of
COMs is a positive multiple of a process.

\subsection{Bipartite States and Composite Systems}

\label{sec:bipart-stat-comp} Given two separate systems, represented
by COMs $A$ and $B$, we should expect that any state of the
composite system $AB$ will induce a joint probability assignment
$p(a,b)$ on pairs of effects $a \in A^{\#}$, $b \in B^{\#}$. If the
two systems can be prepared independently, we should also suppose
that, for any two states $\alpha \in A$ and $\beta \in B$, the
product state $\alpha \otimes \beta$, given by $(\alpha \otimes
\beta)(a,b) = \alpha(a) \beta(b)$, will be a legitimate joint state.
Finally, if the two systems do not interact,
 the choice of measurement made on $A$ ought
not to influence the statistics of measurement outcomes on $B$, and
vice versa. This latter ``no-signaling" condition is equivalent \cite{Wilce92} to
the condition that the joint probability assignment $p$ extends to a bilinear
form on $A^{\#} \times B^{\#}$, normalized so that $p(u_A, u_B) =
1$. Abstractly, then, one makes the following definition.

\begin{definition}\label{defn:bipartite-state}
A (normalized, non-signaling) {\em bipartite state} between convex
operational models $A$ and $B$ is a bilinear form $\omega : A^{\#}
\times B^{\#} \rightarrow {\Bbb R}$ that is {\em positive}, in the
sense that $\omega(a,b) \geq 0$ for all effects $a \in A^{\#}$ and $b
\in B^{\#}$, and \emph{normalized} (satisfies $\omega(u_A,u_B) = 1$).
\end{definition}

Implicit in this definition is the assumption, lately called {\em
local tomography} \cite{D'Ariano2006a}, that a joint state is
determined by the joint probabilities it assigns to measurement
outcomes associated with the local systems $A$ and $B$. As has been
pointed out by many authors, e.g. \cite{Araki, KRF, Barrett}, this
condition is violated in both real and quaternionic quantum theory,
and can therefore be made to serve as an axiom separating standard
complex QM from these. A more general notion of non-signaling
bipartite state would merely associate, rather than identify, each
such state with a positive bilinear form on $A^{\#} \times B^{\#}$.
See the remarks following Definition \ref{defn:composite-state-space}
for more on this.

It is clear that any product $\omega = \alpha \otimes \beta$ of
normalized states $\alpha \in A$ and $\beta \in B$ defines a
non-signaling state; hence, so do convex combinations of product
states. Non-signaling states arising in this way, as mixtures of
product states, are said to be {\em separable} or {\em
  unentangled}. An {\em entangled} non-signaling state is one
that is {\em not} a convex combination of product states.  Many of the
basic properties of entangled {\em quantum} states actually hold for
entangled states in this much more general setting \cite{KRF, BBLW06}.

The space $\B(A^{\#},B^{\#})$ of all bilinear forms on $A^{\#}
\times B^{\#}$, ordered by the cone of all positive bilinear forms,
is the {\em maximal tensor product}, $A \maxtensor B$, of $A$ and
$B$. This notation is reasonable, since (in finite dimensions),
$\B(A^{\#},B^{\#})$ is one model of the tensor product
$(A^{\#})^{\ast} \otimes (B^{\#})^{\ast}$ --- thus, of the
vector-space tensor product $A \otimes B$.\footnote{This is a straightforward extension
  of the definition in \cite{BW09a} to the context of possibly
  non-saturated models.}
Ordering $A \otimes B$ instead by the generally much smaller cone of
unentangled states, that is, the cone generated by the product
states, gives the {\em minimal tensor product}, $A \mintensor B$. It
is important to note that these coincide only when $A$ or $B$ is
classical \cite{BBLW06}. If $A$ and $B$ are quantum state spaces,
then the cone of bipartite density matrices for the composite system
lies properly between the maximal and minimal cones. This indicates
the need for something more general:

\tempout{Although the maximal and minimal tensor products are useful
conceptual tools, they are respectively too large and to small for
many purposes. For example, just as all states in $A \mintensor B$
are unentangled, all effects on $A \maxtensor B$ are unentangled.
Thus, one cannot hope to reproduce quantum information-processing
protocols such as teleportation using only one of these tensor
products; rather, one needs to use both together. }

\begin{definition}\label{defn:composite-state-space}
A (locally tomographic) {\em composite} of COMs $A$ and $B$ is a
convex operational model $(AB, (AB)^{\#},u_{AB})$, such that $AB
\subseteq \B(A^{\#},B^{\#})$, with $u_{AB} = u_{A} \otimes u_{B}$,
$\alpha \otimes \beta \in (AB)_{+}$ for all $\alpha \in A_{+}$,
$\beta \in B_{+}$, and $a \otimes b \in (AB)^{\#}_{+}$ for all $a
\in A^{\#}_{+}, b \in B^{\#}_{+}$.
\end{definition}

It is worth stressing that there are perfectly reasonable theories
that are {\em not} locally tomographic. Indeed, one of these is
quantum mechanics over the real, rather than complex, scalars.  We
might more generally define a {\em composite in the wide sense} of
COMs $A$ and $B$ to be a COM $(AB,(AB)^{\#},u_{AB})$, together with
(i) a positive linear embedding (injection) $i : A \mintensor B
\rightarrow AB$, and (ii) a positive  map $r : AB \rightarrow A
\maxtensor B$, surjective as a \emph{linear} map, such that for all
$a \in A^{\#}, b \in B^{\#}$, \[r(i(\alpha \otimes \beta))(a,b) =
a(\alpha) b(\beta),\] i.e., $r \circ i$ is the canonical embedding
of $A \mintensor B$ in $A \maxtensor B$. However, we shall make no
use of this extra generality here. Accordingly, we assume henceforth
that {\em all composites are locally tomographic}, as per Definition
\ref{defn:composite-state-space} above.

\subsection{Conditioning and Remote Evaluation}

\label{sec:cond-remote-eval} A bipartite state $\omega$ on $A$ and
$B$ gives rise to a positive linear mapping $\hat{\omega} : A^{\#}
\rightarrow B$ with
\[
b(\hat{\omega}(a)) = \omega(a,b)
\]
for all $a \in A^{\#}$
and $b \in B^{\#}$; dually, a bipartite effect $f \in (AB)^{\#}$
gives rise to a linear map $\hat{f}:  A \rightarrow B^{\#}$,
given by $\hat{f}(\alpha)(\beta) = f(\alpha \otimes \beta)$, and
subject to $f(\alpha) \leq u_B$ for all $\alpha \in \Omega_A$.

The \emph{marginals} of $\omega$ are given by $\omega_B =
\hat{\omega}(u_A)$ and $\omega_A = \hat{\omega}^{\ast}(u_B)$. Note that
$\omega$ is normalized iff  $u_{B}(\hat{\omega}(u_A)) = 1$. We can
define the {\em conditional states} of $A$ and $B$ given (respectively) effects $b
\in [0,u_B]$ and $a \in [0,u_A]$ by
\[
\omega_{A|b} := \frac{\homega^*(b)}{\omega_B(b)}
\ \ \text{and} \ \
\omega_{B|a} := \frac{\homega(a)}{\omega_A(a)}
\]
provided the marginal probabilities $\omega_B(b)$ and $\omega_A(a)$
are non-zero. Accordingly, we refer to $\hat{\omega}(a)$ as the {\em
  un-normalized conditional state}.
Notice that the linear adjoint, $\hat{\omega}^{*} : B^{\#} \rightarrow
A$, of $\hat{\omega}$ represents  the same state, but evaluated in the
opposite order: $\hat{\omega}^{*}(b)(a) = \homega(a)(b) = \omega(a,b)$.

\begin{lemma}[Remote Evaluation 1]
\label{lemma:remote-eval-1}
Let $A,B$ and $C$ be convex operational models. For any bipartite effect on $f \in (AB)^{\ast}$ and
any bipartite state $\omega \in BC$, and for any state $\alpha \in
A$,
\[(f \otimes -)(\alpha \otimes \omega) = \homega(\hat{f}(\alpha)).\]
%%Dual version!
\end{lemma}

\noindent{\em Proof:} It is straightforward that this holds where $f$ and
$\omega$ are a product effect and a product state, respectively.
Since these generate $(AB)^{\ast}$ and $AB$, the result follows.
$\Box$\\

Operationally, this says that one can {\em implement} the
transformation $\homega \circ \hat{f}$ by {\em preparing} the tripartite
system $ABC$ in state $\alpha \otimes \omega$, where $\alpha \in A$
is the ``input" state to be processed, and then making a measurement
on $AB$, of which $f$ is a possible outcome: the un-normalized
conditional state of $C$, given the effect $f$ on $AB$, is exactly
$\homega(\hat{f}(\alpha))$.  Thus, the process $\phi := \hat{\omega} \circ \hat{f} : A \rightarrow C$
becomes a special case of {\em conditioning}.
%\footnote{Of course, for this to make
%sense operationally, it is necessary that $f \otimes c$ correspond
%to a {\em bona-fide} effect on the state space $A(BC)$ containing
%the state $\alpha \otimes \omega$. More on this below.}
In
\cite{BBLW08}, we have called this protocol {\em remote evaluation}.
Note that
conclusive, or post-selected, teleportation arises as the special
case of remote evaluation in which, up to some specified
isomorphism, $C \simeq A$ and $\homega \circ \hat{f} \simeq
\id{A}$. We shall return to this point below.

\section{Categories of convex operational models}

\label{sec:mono-prob-theor} We now wish to chart some connections
between the two approaches outlined above. In the first place, we
will bring some category-theoretic order to the concepts developed
in the preceding section. \tempout{Up to a point, this is an
entirely straightforward exercise: we shall consider symmetric
monoidal categories in which the objects are convex operational
models, in which the morphisms are maps thereof, and in which the
the monoidal product of two state spaces is a non-signaling product
of the kind described above. Let us make this precise:}

Since morphisms of COMs compose, we can define a category $\Com$
of all convex operational models and COM morphisms.  As described in
Section \ref{sec:proc-as-posit}, the hom-sets $\Com(A,B)$ are
themselves cones.  Let $I$ denote the COM $\mathbb R$ with its
standard cone and order unit, i.e.  $I = ({\mathbb R},{\mathbb R},
1)$.

\begin{definition}\label{defn:category-of-COMs}
A \emph{category of COMs} is a subcategory
$\C$ of $\Com$ such that:
(i)  $\C(A,B)$ is a (regular) sub-cone of the cone $\Com_{+}(A,B)$;
(ii) $\C$ contains the distinguished COM $I$;
(iii) $\C(I,A) \simeq A$; and
(iv) $\C(A,I) \simeq A^{\#}$.
We call a such a category {\em finite-dimensional} if all state
spaces $A \in \C$ are finite dimensional.
\end{definition}

A more general definition would require only that $\C(A,B)$ be {\em
  some} set of {\em processes}, in the sense of Definition
\ref{defn:process}, between $A$ and $B$.  However, we should like to
be able to construct random mixtures of processes, so $\C(A,B)$ should
at least be convex. Allowing for the taking of limits as a reasonable
idealization, it is plausible to take $\C(A,B)$ also to be
closed. Finally, one should require that, if $\phi : A \rightarrow B$
is a physically valid process, then so is $t\phi$ for any $t \in
[0,1]$ -- this reflecting the possibility of {\em attenuating} a
process (as, for instance, by some filter that admits only a fraction
$t$ of incident systems, but otherwise leaves systems unchanged, or by
in some other way conditioning its occurrence on an event assigned a
probability less than $1$). This much given, the physically meaningful
processes between two systems should generate a closed, convex,
pointed cone of positive mappings, as per Definition 13.

In \cite{BBLW06,BBLW08,BGW09}, a {\em probabilistic theory} is
defined, rather loosely, to be any class of COMs (or ``probabilistic
models") that is equipped with some device, or devices, for forming
composite systems. Tightening this up considerably, we make the following
definition.
%\noindent{\em Remark 2:} We could simply require enrichment over $\Ordlin$. But we won't do so in this paper.

\begin{definition} \label{defn:monoidal-probabilistic-theory}
A {\em monoidal category of COMs} is a category of COMs equipped
with a monoidal structure, such that
(i) the monoidal unit is the COM $I$;
(ii) for every $A, B \in \C$, $A \otimes B$ is a non-signaling composite
in the sense of Definition \ref{defn:composite-state-space}.
%; (more?)
\end{definition}

While Definition
\ref{defn:monoidal-probabilistic-theory} does not require it, in the
rest of the paper we will assume that all monoidal categories of
COMs are \emph{symmetric} monoidal, and (in accordance with our standing assumtion), 
finite dimensional.  

As an example, the category $\FdCom$ of all
finite-dimensional convex operational models and positive mappings can be
made into a monoidal category in two ways, using either the maximal or
the minimal tensor product.  \tempout{(However, as discussed in
  \cite{BW09b}, this is better viewed as a linearly distributive
  category, since it lacks any natural self-dual structure.)}  Another
example is the ``box-world" considered, e.g., in \cite{GMCD10,ShortBarrett09}: here,
state spaces are constructed by forming maximal tensor products of
basic systems, the normalized state spaces of which are
two-dimensional squares.  Another example is afforded by the category
of quantum-mechanical systems, represented as the self-adjoint parts
of complex matrix algebras.  Here, the appropriate monoidal product of
two systems $A$ and $B$ is what is sometimes referred to as the
``spatial" tensor product, obtained by forming tensor products of the
Hilbert spaces on which the $A$ and $B$ act, and taking the
self-adjoint operators on this space.

\subsection{Remote Evaluation Again}

\label{sec:remote-eval-again} We now reformulate the  conditioning
maps and remote evaluation protocol discussed above in purely
categorical terms.  In fact, both make sense in any symmetric
monoidal category $\C$. Suppose, then, that $ \omega : I \rightarrow
B \otimes A$ is a ``bipartite state", i.e, a state of the composite
system $B \otimes A$. Then there is a canonical mapping
$\hat{\omega} : \C(B,I) \rightarrow \C(I,A)$ given by

\begin{equation}\label{eq:remote-eval-again}
{\xymatrix@=12pt{ I \ar@{->}^{\omega}[rr] \ar@{->}_{\hat{\omega}(b)}[ddrr] & & B \otimes A \ar@{->}[dd]^{b \otimes \id{A}}\\
& & \\
& & A }}
\end{equation}
Dually, if $f  \in \C(A \otimes B, I)$, there is a
natural mapping $\hat{f} : \C(I,A) \rightarrow \C(B,I)$ given by
\begin{equation}\label{eq:remote-eval-again-dual}
{\xymatrix@=12pt{ B \ar@{->}^{\alpha \otimes \id{B}}[rr] \ar@{->}_{\hat{f}(\alpha)}[ddrr] & & A \otimes B \ar@{->}[dd]^{f}\\
& & \\
& & I }}.
\end{equation}
Note that if $\C$ is {\em already} a category of COMs, then
$\hat{\omega}$ and $\hat{f}$ are exactly the conditioning and
co-conditioning maps discussed in the last section. This has a
simple but important corollary, namely, that these maps are
indeed morphisms of COMs. Another consequence is that any monoidal
category  of COMs is {\em closed under conditioning} --
that is, if $\omega$ is a normalized bipartite state of such a
theory, belonging, say, to a composite system $AB$, then for every
effect $a$ on $A$ and $b$ on $B$, the composite states
$\omega_{B|a}$ and $\omega_{A|b}$ are indeed states of $A$ and $B$,
respectively (as opposed to merely being elements of $(A^{\#})^{\ast}$ and
$(B^{\#})^{\ast}$).\footnote{This is closely related to the notion of \emph{regular
composite} introduced in \cite{BBLW08}.}

The remote evaluation protocol of Lemma \ref{lemma:remote-eval-1} also
has a purely category-theoretic formulation:

\begin{lemma}[Remote Evaluation 2]
\label{lemma:remote-eval-2}
Let $\omega : I \rightarrow B \otimes C$ and $f : A \otimes B \rightarrow I$ in $\C$. Then
\begin{equation} \label{eq:remote-eval-ii}
\hat{\omega}(\hat{f}(\alpha)) \
 = \ (f \otimes \id{C})\circ (\id{A} \otimes \omega) \circ \alpha
 = (f \otimes \id{C}) \circ (\alpha \otimes \omega)
\end{equation} for all
$\alpha \in \C(I,A)$. Dually, for every $\beta \in \C(I,B)$, we have
\begin{equation}\label{eq:remote-eval-ii-bis}
\hat{\omega}^{\ast}(\hat{f}^{\ast}(\beta)) \
= \ (\id{A} \otimes f) \circ (\omega \otimes \beta).
\end{equation}
\end{lemma}

\noindent{\em Proof}: We prove \eqref{eq:remote-eval-ii}, the proof of \eqref{eq:remote-eval-ii-bis} being
similar. Tensoring the diagram  \eqref{eq:remote-eval-again-dual} with $C$
(on the right) gives the right-hand triangle in the diagram below. Applying (4) to compute $\hat{\omega}(\hat{f}(\alpha))$ gives the
lower triangle. The square commutes by the bifunctoriality of the tensor.
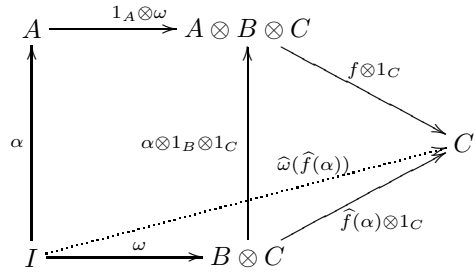
\[\ \
{\xymatrix@=12pt{ A \ar@{->}^{\id{A} \otimes \omega}[rrr]  & & & A \otimes B \otimes C \ar@{->}^{f \otimes \id{C}}[rrrdd] & & & & &  \\
& & & & & & & & \\
& & & & & & C & & \\
& & & &  & & & &  \\
I \ar@{->}[uuuu]^{\alpha} \ar@{->}^{\omega}[rrr]
\ar@{..>}^{\rule[-2ex]{0mm}{2ex}\qquad\qquad\qquad\qquad\hat{\omega}(\hat{f}(\alpha))}[rrrrrruu] & & & B \otimes
C \ar@{->}^{\alpha \otimes \id{B} \otimes \id{C}}[uuuu]
\ar@{->}_{\hat{f}(\alpha) \otimes \id{C} }[uurrr]
 & & & & & }}\]
Chasing around the diagram gives the desired result.
\rule{2mm}{0mm}\qed

Suppose that, in the preceding lemma, $\omega(1) \in A \otimes B$ is
a normalized state, and $f : B \otimes C \rightarrow I$ is an
effect, i.e, $0 \leq f \leq u_{BC}$. Then, in operational terms, the
Lemma says that the mapping $\hat{\omega} \circ \hat{f}$ is
represented, within the category $\C$, by the composite morphism $(f
\otimes \id{C}) \circ (\id{A} \otimes \omega)$.\footnote{
  Technically we are relying on the isomorphisms between $A \iso \C
  (I,A)$ and $A^\# \iso \C (A,I)$ to guarantee that the internal
  representation of $\homega  \circ \hat{f}$ defines the right linear
  map.
}
In other words: preparing $BC$
in joint state $\omega$, and then measuring $AB$ and obtaining $f$,
guarantees that the ``un-normalized conditional state" of $C$ is
$\hat{\omega}(\hat{f}(\alpha))$, where $\alpha$ is the state of $A$.

%\noindent{\em Remark:}
\begin{remark}\label{remark:classicalcond-not-statecollapse}
  An important point here is that any process that factors as
  $\hat{\omega} \circ \hat{f}$ can be simulated by a remote evaluation protocol,
  using what amounts to {\em classical conditioning} -- in particular,
  without need to invoke any mysterious ``collapse" of the state,
  nor for that matter, any other {\em physical} dynamics at all.
\end{remark}

\subsection{Teleportation, conditional dynamics and compact closure}

\label{sec:telep-cond-dynam} Suppose that, in the remote evaluation
protocol of Lemma \ref{lemma:remote-eval-2}, we have $C = A$.
Suppose further that the mapping $\hat{\omega} : B^{\#} \rightarrow
A$ has a right inverse --- that is, suppose there exists a positive
linear map $\hat{r}:A \to B^\#$ such that $\homega \circ \hat{r} =
\id{A}$.   Then we can re-scale $\hat{r}$ to obtain an effect $f$ on
$A \otimes B$ by
\[
f(\alpha , \beta ) = c \hat{r}(\alpha)(\beta),
\]
for a small enough positive constant $c$.  Upon obtaining the result $f$ in
a measurement on $A \otimes B$ when the composite system is in state
$\alpha \otimes \omega$, the un-normalized conditional state of $C$
is:
\[
\hat{\omega}(\hat{f}(\alpha)) = c \alpha.
\]
The {\em normalized} conditional state will be exactly $\alpha$.  This
is what is meant, in quantum-information theory, by a conclusive,
correction-free {\em teleportation protocol}. Adopting this language,
we will say that it is possible to {\em teleport system $A$ through
  system $B$} if and only if there exists such a pair $\homega,
\hat{r}$.

If $\homega$ is in fact an isomorphism $A^\# \iso B$, then $\hat{r} =
\homega^{-1}$, and system $B$ can also be teleported through system $A$.
When this is the case, Lemma \ref{lemma:remote-eval-2}
tells us that $\omega : I \rightarrow B \otimes A$ and $ f : A \otimes
B \rightarrow I$ with $\hat{f} = \hat{\omega}^{-1}$, provide
respectively a unit and co-unit making $(B,\omega,f)$ a dual for
$A$. Thus, a compact closed category of COMs is exactly one in which
every system $A$ is paired with a second system $B = A'$, in such a way that
each system can be teleported through the other.

\begin{proposition} \label{prop:comcl-is-equiv-to-teleportation}
Let $\C$ be a monoidal category of COMs. The
  following are equivalent.
\begin{itemize}
\item[(a)] $\C$ is compact closed.
\item[(b)] Every $A \in \C$ can be teleported through some $B \in \C$,
  which in turn can be teleported through $A$.
\item[(c)] Every morphism in $\C$ has the form $\hat{\omega} \circ
  \hat{f}$ for some bipartite state $\omega$ and bipartite effect $f$.
\end{itemize}
\end{proposition}

\noindent{\em Proof:} The equivalence of (a) and (b) is clear from the
preceding discussion. To see that these are in turn equivalent to (c),
suppose first that (a) holds, and let $(A', \eta_A, e_A)$ be
the dual for $A$.   Suppose that  $\phi : A\to B$ is a morphism in
$\C$, and define $\omega_\phi = (\id{A'}\otimes \phi) \circ \eta_A.$  By
Remote Evaluation (Lemma \ref{lemma:remote-eval-2}), we have
\[
\hat{\omega_\phi}(\hat{e_A}(\alpha)) \
 = \ (e_A \otimes \id{B})\circ (\id{A} \otimes \omega_\phi) \circ \alpha
\]
for every $\alpha \in \C(I,A)$.  Since $\C$ is compact closed, the
following diagram commutes:
\[{\xymatrix@=12pt{
A\ar@{->}^{\id{A}}[rr]\ar@{->}_{\id{A}\otimes \eta_A}[rrdd] && A\ar@{->}^{\phi}[rr] && B\\
\\
&& A \otimes A' \otimes A \ar@{->}^{\;\id{A}\otimes \id{A'} \otimes
  \phi\;}[rr] \ar@{->}_{e_A \otimes \id{A}}[uu] && A \otimes A'
\otimes B\;, \ar@{->}_{e_A \otimes \id{B}}[uu]
}}\]
and hence $\hat{\omega_\phi}(\hat{e_A}(\alpha)) = \phi(\alpha)$.
Since $\C(I,A) \iso A$ we have $\phi = \hat{\omega_\phi} \circ
\hat{e_A}$ as required. Conversely, if (c) holds, then for each $A$, the
identity mapping $\id{A}$ factors as $\hat{\omega}_{A} \circ
\hat{f}_{A}$ for some $\omega_{A} \in B \otimes A$ and some $f \in A
\otimes B$. It follows that $\hat{\omega}_{A} = \hat{f}_{A}^{-1}$, so
this gives us a compact closed structure. $\Box$ \\

\section{Weakly Self-Dual Theories}
\label{sec:weakly-self-dual}

In a compact closed category $\C$, the internal adjoint $' : \C
\rightarrow \C$ described in Section 2 establishes an isomorphism
$\C \simeq \C^{\op}$. In particular, for every object $A$ in the
category, understood as a ``physical system", there is a distinguished
isomorphism between the system's state space $\C(I,A)$ and the space
$\C(A,I)$ of effects.

In contrast, a convex operational model $A$ is not generally
isomorphic to its dual. Indeed, there is a type issue: $A$ has, by
definition, a distinguished unit functional $u_A \in A^{\#}$; in order
for $A^{\#}$ to be treated as a COM, one would need to privilege a
{\em state} $\alpha_o \in A$ to serve as an order unit on $A^{\#}$.
Only in special cases is there a natural way of doing so.\footnote{When the
state space is sufficiently symmetric, there may be a natural choice of
state invariant under the symmetry group.  For example, if the base-preserving
automorphisms act transitively on the pure states, the state obtained by
group-averaging is the natural choice.}
Beyond this, there is the more fundamental problem that,
geometrically, the cones $A_{+}$ and $A^{\#}_{+}$ are generally not
isomorphic.  This said, those convex operational models that {\em are}
order-isomorphic to their duals are of considerable interest -- not
only because both classical and quantum systems exhibit this sort of
self-duality, but because it appears to be a strong constraint, in
some measure {\em characteristic} of these theories.
%%%This last maybe too strong...

\subsection{Weak vs Strong Self-Duality}

A finite-dimensional ordered vector space $A$ (or its cone, $A_{+}$)
is said to be {\em self-dual} iff there exists an inner product --
that is, a positive-definite, hence also symmetric and
non-degenerate, bilinear form $\langle \ , \ \rangle$ -- on $A$ such
that
\[A_{+} = A^{+} := \{ a \in A | \langle a, x \rangle \geq 0 \ \forall x \in A_{+}\}.\]
In this case, we have $A \simeq A^{\ast}$, as ordered spaces, via the
canonical isomorphism $a \mapsto \langle a , . \rangle$.  A celebrated
theorem of Koecher and Vinberg \cite{Koecher58,Vinberg,FK} states
that if $A_{+}$ is both self-dual and {\em homogeneous}, meaning that
the group of order-automorphisms of $A$ acts transitively on the
interior of $A_{+}$, then $A$ is isomorphic, as an ordered space, to a
formally real Jordan algebra ordered by its cone of squares.  The
Jordan-von Neumann-Wigner \cite{JvNW} classification of such algebras
then puts us within hailing distance of quantum mechanics.

\begin{definition}\label{def:wsd-asp} A COM $(A,A^{\#},u_A)$ is {\em weakly self-dual}
(WSD) iff there exists an order-isomorphism $\phi : A \simeq
A^{\#}$. We shall say that $A$ is {\em symmetrically} self-dual iff $\phi$ can be so chosen that
$\phi(\alpha)(\beta) = \phi(\beta)(\alpha)$
for all $\alpha, \beta \in A$. \end{definition}

Note that, for a given linear map $\phi : A \rightarrow A^{\#}$, the
bilinear form $\langle \alpha, \beta \rangle := \phi(\alpha)(\beta)$
is non-degenerate iff $\phi$ is a linear isomorphism, and symmetric
iff $\phi = \phi^{\ast}$.  If we don't require saturation, {\em any}
finite-dimensional ordered linear space $A$ can serve as the state
space for a weakly self-dual COM, simply by setting $(A^{\#})_{+} =
\phi(A_+)$ for some nonsingular positive linear mapping $A
\rightarrow A^{\ast}$, and taking any point in the interior of
$(A^{\#})_+$ for $u_A$.  However in the saturated case, weak
self-duality is a real constraint on the geometry of the state cone,
although strong self-duality is an even stronger one.

\tempout{
\begin{lemma} Let $(A,A^{\#},u_A)$ be a finite-dimensional, symmetrically self-dual COM. Then $A^{\#} = A^{\ast}$, i.e.,
the model is saturated.\end{lemma}

\noindent{\em Proof:} Let $\phi : A \rightarrow A^{\#}$ be an
order-isomorphism defining a symmetric bilinear form $\langle \alpha ,
\beta \rangle := \beta(\phi(\alpha)) = \alpha(\phi(\beta))$. Then we
have Let
\[A^{+} = \{ \beta \in A | \langle \alpha, \beta \rangle \geq 0 \forall \alpha \in A_{+}\}.\]
Since $\phi$ takes $A_{+}$ onto $A^{\#}_{+}$, $\beta \in A^{+}$ iff
$a(\beta) \geq 0$ for all $a \in A^{\#}_{+}$ -- in other words,
$A^{+} = (A^{\#})^{\ast}_{+}$ (here exploiting the
finite-dimensionality of $A$ to identify $A$ with $A^{\ast \ast}$).
Now, since $\langle \ , \ \rangle$ is symmetric, we also have $\beta
\in A^{+}$ iff $\phi(\beta) \in A^{\#}_{+}$, whence, $A^{+} =
\phi^{-1}(A^{\#}_{+}) = A_{+}$. An application of the
(finite-dimensional) Hahn-Banach theorem now tells us that that
$A^{\ast} = A^{\#}$, i.e., that $A$ is saturated. $\Box$\\}

Note that, for a given linear map $\phi : A \rightarrow A^{\#}$, the
bilinear form $\langle \alpha, \beta \rangle := \phi(\alpha)(\beta)$
is non-degenerate iff $\phi$ is a linear isomorphism,  and symmetric
iff $\phi = \phi^{\ast}$. Thus, $A$ will be self-dual, in the
classical sense described above, iff $\langle \ , \ \rangle$ is
positive-definite, {\em and} $A^{\#} = A^{\ast}$, i.e, $A$ is
saturated. To emphasize the distinction, we shall henceforth refer
to this situation as {\em strong} self-duality.

If $\phi : A \simeq A^{\#}$  is an order-isomorphism implementing $A$'s
self-duality, then $\phi^{-1} : A^{\#} \simeq A$ defines an un-normalized bipartite
non-signaling state $\gamma$ in $A \maxtensor A$ with $\phi^{-1} = \hat{\gamma}$
-- that is, $\gamma(a,b) = \phi^{-1}(a)(b)$. Following \cite{BGW09},
we shall call such a state an {\em isomorphism state}. It is shown
in \cite{BGW09} that such a state is necessarily pure in $A
\maxtensor B$.\footnote{%
  Strictly speaking, \cite{BGW09} deals with
  the case in which $A$ and $B$ are saturated, but the proof is easily
  extended to the general case.}
In this language, {\em $A$ is
WSD iff $A \maxtensor B$ contains an isomorphism state}.

 \begin{example}
Let $A$ be the convex operational model of a basic quantum-mechanical
system, i.e., the space of self-adjoint operators associated with the
system's Hilbert space $\H$. The standard {\em maximally entangled state}
on $A \otimes A$ is the pure state associated with the unit vector
\[
\Psi = \frac{1}{\sqrt{d}}\sum_{i} x_i \otimes x_i
\]
where $\{x_1,...,x_n\}$ is an orthonormal basis for $\H$. Using this,
one has a mapping
\[
R : T \mapsto R_{T} := (T \otimes \1)P_{\Psi}
\]
taking operators $T: \H \rightarrow \H$ to operators ${\cal B}(\H
\otimes \H)$. It is a basic result, due to Choi and, independently,
Jamiolkowski, that this is a linear isomorphism, taking the cone of
completely positive maps on ${\cal B}(\H)$ onto the cone of positive
operators on $\H \otimes \H$. Note that $R^{-1}$ maps
\cite{Keyl-Werner} $\rho$ to $T_{\rho}$ where the latter is given by
\[
 \langle x | T_{\rho} (\sigma) y\rangle
 = d \Tr(\rho ((\ket{y}\bra{x}) \otimes \sigma^T))
\]
where $\ket{y}\bra{x}$ is the operator $z \mapsto \langle z, x \rangle y$,
and the transpose is defined relative to the chosen orthonormal
basis. This gives us a state $\gamma \in A \otimes A$ with
$\hat{\gamma} : A^{\ast} \simeq A$, namely,
\[
\hat{\gamma}(a)(b)
= \Tr(P_{\Psi}(a \otimes b))
= \Tr(P_{\Psi}(a \otimes \1) (\1 \otimes b)).
\]
\end{example}

\subsection{Categories of self-dual COMs}

\label{sec:categories-self-dual} Let $\C$ be a monoidal category of
COMs, as described in Section \ref{sec:mono-prob-theor}. There is a
distinction between requiring that a state space $A \in \C$ be
weakly self-dual, which
implies only that there exist an order-isomorphism $A^{\#} \simeq A$    
-- an isomorphism in $\Ordlin$ -- and requiring that this correspond
to an (un-normalized) state $\gamma \in (AA)_+$, hence, to an element
of $\C(I,A \otimes A)$.  We now focus on categories in which this latter condition
holds for every system.

\begin{definition} \label{def:wsd-cat} A symmetric monoidal category
$\C$ of COMs is {\em weakly self-dual} (WSD) iff for every $A \in
\C$, there
  exists a pair $(\gamma_A, f_A)$ consisting of a bipartite state
  $\gamma_A \in A \otimes A$ and a positive functional $f_A \in
  (A\otimes A)^{\#} = \C(A \otimes A,I)$ (a multiple of an effect)
  such that (i) $\gamma_A$ is an isomorphism state, and (ii)
  $\hat{f}_A = \hat{\gamma}_{A}^{-1}$.  If $\gamma_A$ can be chosen to
  be symmetric for every $A \in \C$, we shall say that $\C$ is {\em symmetrically self-dual} (SSD).
\end{definition}

Note that this is stronger than merely requiring every COM $A \in
\C$ to be weakly self-dual.  Equivalently, we may say that category
$\C$ of COMs is WSD iff every system can be equipped with a
designated conclusive, correction-free teleportation protocol
$\gamma_A \in AA, f_A \in (AA)^{\#}$, whereby $A$ can be teleported
``through itself''. Thus, by Proposition
\ref{prop:comcl-is-equiv-to-teleportation}, we have:

\begin{theorem} \label{cor:wsd-cat-compact-closed}
A monoidal category $\C$ of convex operational models is weakly
self-dual iff it is compact closed, and can be equipped with a compact
structure such that $A' = A$ for all objects $A \in \C$.
\end{theorem}

Recall that any morphism $\phi$ in a compact closed
category has a categorial adjoint, $\phi'$.  In the context of a WSD category
of COMs, this has a useful interpretation in terms of the linear adjoint, $\phi^{\ast}$:

\begin{lemma} \label{lem:wsd-canonical-adjoint}
Let $\C$ be any WSD category of ASPs, regarded as compact closed as
above. Let $\phi : A \rightarrow B$. Then the canonical adjoint
mapping $\phi' : B \rightarrow A$ is given by
\[
\phi' = (\hat{f}_{B} \circ \phi \circ \hat{\gamma}_{A})^*
      = \hat{\gamma}_{A}^* \circ \phi^* \circ \hat{f}_{B}^*
\]
\end{lemma}

\begin{proof}
  We must show that, for any $\beta \in B$ -- that is, any $\beta \in
  \C(I,B)$ -- we have $\phi'(\beta) := \phi^{*} \circ \beta =
  \hat{\gamma}_{A}^*(\phi^{*}(\hat{f}_{B}^*(\beta)))$, where $\phi^* :
  B^{\#} \rightarrow A^{\#}$ is the {\em linear} adjoint. Let $\omega
  := (\id{A} \otimes \phi) \circ \gamma_A : I \rightarrow A \otimes
  B$. Then we have
  \begin{eqnarray*}
    \phi'(\beta) & = & \phi' \circ \beta \\
    & = &  (\id{A} \otimes f_B) \circ (\id{A} \otimes \phi \otimes \id{B}) \circ (\gamma_A \otimes \id{B}) \circ \beta\\
    & = & (\id{A} \otimes f_B) \circ (\id{A} \otimes \phi \otimes
    \beta) \circ \gamma_A \\
    % & = & (\id{A} \otimes f_B) \circ ((\id{A} \otimes \phi)\circ \gamma \otimes \beta)\\
    & = & (\id{A} \otimes (f_B \circ (\id{B}\otimes \beta))) \circ
    ((\id{A} \otimes \phi) \circ \gamma_A) \\
    & = & (\id{A} \otimes \hat{f}^* (\beta))  \circ \omega )\\
    & = & \hat{\omega}^{*}(\hat{f}^{*}_B(\beta)).
  \end{eqnarray*}
  Now, for any $b : B \rightarrow I$, we have
  \begin{eqnarray*}
    \hat{\omega}^{*}(b) & = & \widehat{(\sigma_{B,B} \circ \omega)}(b)\\
    & = & (\id{A} \otimes b) \circ \omega \\
    & = & (\id{A} \otimes b) \circ (\id{A} \otimes \phi) \circ \gamma_A\\
    & = & (\id{A} \otimes (b \circ \phi)) \circ \gamma_A \\
    & = & (\id{A} \otimes \phi^{*}(b)) \circ \gamma_A\\
    & = & \hat{\gamma}^{*}_A(\phi^{*}(b)).
  \end{eqnarray*}
  With $b = \hat{f}^{*}_B(\beta)$, this gives the desired result.
\end{proof}

\begin{corollary} \label{cor:unit-counit-symmetry}
For all $A \in \C$, $f_{A}' = \sigma_{A,A} \circ \gamma_{A}$.
\end{corollary}

\begin{proof}
  Note first that $f^{*}_A(1) = f_A \in (A \tensor A)^{\#}$. Thus, the
  preceding Lemma gives us
  \[
  f_{A}'(1) = (\hgamma_{A \otimes A}^{*} \circ f_{A}^{*} \circ
  \hf^{*}_{I})(1)
  \]
  Since $f_{I} = f_{I}^{*} = \id{I}$, we have
  \[
  f_{A}'(1) = (\hgamma_{A}^{*} \otimes \hgamma_{A}^{*})(f_A).
  \]
  Thus, for every $a, b \in A^{\#}$, we have
  \begin{eqnarray*}
    f_{A}'(1)(a,b)
    & = & (\hgamma_{A}^{*} \otimes \hgamma_{A}^{*})(f_{A})(a,b) \\
    & = & f_{A}(\hgamma_{A}(a),\hgamma_{A}(b)) \\
    & = & \hf_{A}^{*}(\hgamma_{A}(b))(\hgamma^{*}_{A}(a)) \\
    & = & b (\hgamma_{A}^{*}(a)) \\
    & = & \gamma_{A}(b,a) \\
    & = & (\sigma \circ \gamma_{A})(a,b).
  \end{eqnarray*}
\end{proof}

\begin{theorem}\label{thm:symmetryWSD-daggercompact}
Let $\C$ be a {\em symmetrically} self-dual monoidal category of COMs.
Then $\C$ is dagger compact with $\dagger$ given by the canonical
adjoint $' : \C^{\op} \rightarrow \C$.
\end{theorem}
\begin{proof}
  By Theorem \ref{cor:wsd-cat-compact-closed}, $\C$ is degenerate
  compact closed, with a compact structure on $A \in \C$ given by $(A,
  \gamma_A, f_A)$, where $\gamma_A$ is symmetric.  Define
  $(\cdot)^\dag : \C^{\textrm{op}} \to \C$ by $(\cdot)^\dag =
  (\cdot)'$; then $(\cdot)^\dag$ is a monoidal functor which is the
  identity on objects.
  \tempout{ and $\gamma_A$ and $f_A$ are symmetric by
  Proposition \ref{prop:involutive-implies-symmetry}.} %%AW: what's this about?
   Since $\sigma_{A,B}'
  = \sigma_{A,B}^{-1}$ in any compact closed category, $\C$ is
  dagger-monoidal.   That $f_A^\dag
  = \sigma_{A,A} \circ \gamma_A$ is immediate from Corollary
  \ref{cor:unit-counit-symmetry}.
\end{proof}

\begin{lemma}\label{lem:double-dual}
For all $A$, let $\tau_A : A \rightarrow A$ be the
order-isomorphism given by
\[
\tau_{A} := \hgamma_{A} \circ \hf^{*}_{A}.
\]
Then, for all $\phi \in \C(A,B)$, we have
\[
\phi'' = \tau_{B}^{-1} \circ \phi \circ \tau_{A}.
\]
\end{lemma}

\noindent{\em Proof:} Notice first that $\tau_{B}^{-1} =
(\hgamma_{B} \circ \hf^{*}_{B})^{-1} = \hgamma_{B}^{*} \circ
\hf_{B}$. By Lemma \ref{lem:wsd-canonical-adjoint}, we have
\begin{eqnarray*}
\ \ \phi'' = (\hf_{A} \circ \phi' \circ \hgamma_B)^{*} & = & (\hf_{A} \circ (\hf_B \circ \phi \circ \hgamma_A)^{*} \circ \hgamma_{B})^{*} \\
& = & \hgamma_{B}^{*} \circ (\hf_{B} \circ \phi \circ \hgamma_{A})^{**} \circ \hf_{A}^{*} \\
& = & (\hgamma_{B}^{*} \circ \hf_{B}) \circ \phi \circ (\hgamma_{A} \circ \hf_{A}^{*})\\
& = & \tau_{B}^{-1} \circ \phi \circ \tau_{A} \ \ \Box
\end{eqnarray*}

\begin{corollary}\label{cor:centre}
Let $\phi : A \rightarrow B$ with $\phi'' = \phi$. Then
$\phi \circ \tau_{A} = \tau_{B} \circ \phi$.
\end{corollary}

\begin{theorem} \label{prop:involutive-implies-symmetry}
For any object $A$ in a weakly self-dual category $\C$ of
convex operational models, the following are
equivalent: (i) $\phi'' = \phi$ for all $\phi
\in \C(A,A)$,
%%Suppose,  too, that the center of $G_{\C}(A_+)$ consists
%%entirely of scalar multiples of the identity.
(ii) $\tau_{A} =
\id{A}$, and (iii) $f_A$ and $\gamma_A$ are  symmetric as a bilinear forms.
\end{theorem}

\noindent{\em Proof:}
(i) implies (ii): From (i), and the fact that the morphisms in $\C(A,B)$ are a
basis for $\cl(A,B)$, $(\cdot)''$ is the identity map.  Let $E_i, F_j$
be bases of the spaces $\cl(A,A), \cl(B,B)$ of linear maps on vector
spaces $A,B$ respectively.  Then the maps $X \mapsto F_j X E_i$, where
$X \in \cl(A,B)$, are a basis for the space of linear maps from
$\cl(A,B)$ to itself.  Using this fact, we can expand the map
$(\cdot)'': \phi \mapsto \tau_B^{-1} \circ \phi \circ \tau_A$, which
is a map from $\cl(A,B)$ to itself, in a basis $M_{ij}: \phi \mapsto
F_j X E_i$ where $E_0= 1_B, F_0=1_A$.  By the uniqueness of expansions
in bases and the fact that $(\cdot)''$ is the identity map, we get
$\tau_A = 1_A, \tau_B = 1_B$.  \\ (ii) implies (iii): By (ii) we have $\tau_A :=
\hgamma_A \circ \hf^*_A = 1_A$, so $\hf^*_A = \hgamma_A^{-1} = \hf_A$.
Since $f(a,b) \equiv \hf(a)(b) \equiv \hf^*(b)(a)$, $\hf^*_A = \hf_A$
is equivalent to symmetry of $f_A$.  Symmetry of $\gamma_A$ then
follows from the fact that $\gamma_A = f_A^{-1}$.\\
(iii) implies (i): If $f_A$ (hence, also $\gamma_A$) are symmetric, then we have
$f_{A}^{\ast} = f_{A}$, whence, $\tau_{A} = \gamma_A \circ \hf_{A} = 1_A$;
thus, by Lemma \ref{lem:double-dual}, $\phi'' = \phi$ for all $\phi \in \C(A,A)$. $\Box$\\

Applying Theorem 24, we now have the

\begin{corollary}\label{thm:symmteryWSD-daggercompact}
A WSD monoidal category of COMs is dagger compact with
respect to the canonical adjoint $' : \C^{\op} \rightarrow \C$, if and
only if it is {\em symmetrically} self-dual. \end{corollary}

Theorem~\ref{prop:involutive-implies-symmetry} tells us that if $\C$ is
a \emph{saturated} weakly self-dual theory in which $'$ is an
involution, then the interior of $A_+$ is a domain of positivity in the sense of
Koecher \cite{KoecherNotes}.  If we further suppose that every
irreducible state space in $\C$ is {\em homogeneous}, meaning that
$G(A_+)$ acts transitively on the interior of $A_{+}$ for every $A \in
\C$ (a condition one can motivate physically in several ways, e.g.,
\cite{BGW09, Wilce09}), then we are close to requiring that every
state space in $\C$ be a formally real (also called Euclidean)
Jordan algebra \cite{Koecher58, Vinberg}. This line of thought will be
pursued in a sequel to this paper.

\section{Conclusion}
\label{sec:conclusion}

We began with the observation that compact closure, and still more,
dagger compactness, represent strong constraints on a physical
theory, thought of as a symmetric monoidal category of ``systems"
and ``processes". Our results cast some light on the operational
(or, if one prefers, the physical) content of these assumptions in
the concrete --- but still very general --- context of probabilistic
(or convex operational) theories. In particular, we have seen that,
in probabilistic theories {\em qua} categories of COMs, compact
closure amounts to the condition that all processes -- that is, all
dynamics -- can be induced by the kind of conditioning that occurs
in a teleportation-like protocol. Indeed, in such a theory, a
process between systems $A$ and $B$ {\em amounts to} a choice of
bipartite state on $A \otimes B$.  Finally, we have established that
for weakly self-dual theories \emph{symmetric} weak self-duality
implies the existence of a dagger compatible with the compact closed
structure.  As in the special case of quantum mechanics, a dagger
amounts to reversing the order of conditioning.

Throwing these results into sharper relief is the following result
(details of which will appear elsewhere).  Let $\C$ be any symmetric
monoidal category, with associated monoid of scalars $S = \C(I,I)$.
Then, for any monoid homomorphism $p : S \rightarrow [0,1]$, there
exists a functor $V_p : \C \rightarrow \Com$. This can be used to
transfer the monoidal structure of $\C$ to the image category
$F(\C)$, giving us a representation of $\C$ {\em as} a category of
COMs. The construction is straightforward: one defines a Mackey
triple $(X_A,\Sigma_A,p_A)$ for each object $A \in \C$, with $X_A =
\C(A,I)$, $\Sigma_A = \C(I,A)$ and $p_A(x,\alpha) = p(x \circ
\alpha)$; linearizing this, as in Example \ref{ex:mackey-triple},
yields a COM $(V(A),V(A)^{\#},u_A)$.

Several directions for further study suggest themselves.  It would be
interesting to identify necessary and sufficient conditions for the
COM representations discussed in section 4.1 to yield {\em finite
  dimensional} models -- and, equally, one would like to know how far
the other results of Sections 4 and 5 extend to infinite-dimensional
systems.  Our definition of category of weakly self-dual state
spaces assumes the existence of a state that induces, by
conditioning, an isomorphism from the state cone to the effect cone,
and an effect inducing its inverse; but as we noted, it would be
interesting to investigate conditions under which this follows just
from weak self-duality of the objects.  As mentioned at the end of
section \ref{sec:weakly-self-dual}, the consequences of homogeneity
of the state-spaces should also be explored.  Perhaps the most
urgent task, though, is to identify operational and
category-theoretic conditions equivalent to the strong self-duality
of a probabilistic theory.


\begin{thebibliography}{}

\bibitem{Abramsky-Coecke} S. Abramsky and B. Coecke, A categorical semantics of quantum
protocols, in \emph{Proceedings of the 19th Annual IEEE Symposium on Logic in
Computer Science: LICS 2004}, pp. 415--425, IEEE Computer Society (2004); also
 quant-ph/0402130v5 (2004).
\bibitem{ACHbk} S. Abramsky and B. Coecke, Categorical quantum mechanics, in
\emph{Handbook of Quantum Logic and Quantum Structures} {\bf II}, K. Engesser,
D. Gabbay, D. Lehman, eds.,
Elsevier (2008).
\bibitem{AD} S. Abramsky and R. Duncan, A categorical quantum logic,
Mathematical Structures in Computer Science {\bf 16}, 486--489 (2006).
\bibitem{Alfsen} E. M. Alfsen, {\em Compact Convex Sets and Boundary
  Integrals}, Springer, 1971.
\bibitem{AlfsenShultzBook} E.~M. Alfsen and F.~W. Shultz, {\em State Spaces of
Operator Algebras: Basic Theory, Orientations, and {$C^*$}-products.}
Boston, Birkhauser (2002).
\bibitem{Araki} H. Araki, 1980, On a characterization of the state
space of quantum mechanics, Comm. Math. Phys {\bf 75} (1980) 1-24.
\bibitem{Baez04a}
J.~Baez.
\newblock Quantum quandaries: a category-theoretic perspective.
\newblock {\tt quant-ph/0404040}, 2004.
\bibitem{BBLW06} H. Barnum, J. Barrett, M. Leifer and A. Wilce,
Generalized no-broadcasting theorem, Phys. Rev. Lett. {\bf 99}
(2007) 24051
\bibitem{BBLW08} H. Barnum, J. Barrett, M. Leifer and A. Wilce,
  Teleportation in general probabilistic theories, in Proceedings of the
  Clifford Lectures, Tulane University, March 12-15, 2008, to appear in
\emph{Proceedings of Symposia in Applied Mathematics} (AMS);
  also arXiv:0805.3553 (2008).
\bibitem{BDLT2008} H. Barnum and O. Dahlsten and M. Leifer and
  B. Toner, Nonclassicality without entanglement enables bit
  commitment, \emph{Proc. IEEE Information Theory Workshop, Porto, May 2008},
  arXiv:0803.1264 (2008).
\bibitem{BFRW} H. Barnum, C. Fuchs, J. Renes and A. Wilce,
Influence-free states on compound quantum systems, quant-ph/0507108
(2005)
\bibitem{BGW09} H. Barnum, P. Gaebler and A. Wilce, Ensemble
steering, weak self-duality and the structure of probabilistic
theories, {\tt arXiv:0912.5532} (2009)
\bibitem{BW09a} H.~Barnum and A.~Wilce, Information processing in convex
  operational theories (2009) (arXiv:0908.2352).  To appear in a
  special issue of {\em Electronic Notes in Theoretical Computer
    Science}: Proceedings of QPL/DCM (Quantum Physics and Logic /
  Developments in Computational Models), Reykjavik, July 12-13, 2008.
\bibitem{BW09b} H.~Barnum and A.~Wilce, Ordered linear spaces and
  categories as frameworks for information-processing
  characterizations of quantum and classical theory (2009); also arXiv:0908.2354.
\bibitem{Barrett} J.~Barrett, Information processing in general probabilistic
theories, \emph{Phys. Rev. A} {\bf 75} 032304 (2007); also arXiv:quant-ph/0508211.
\bibitem{Beltrametti-Cassinelli} E. Beltrametti and G. Cassinelli, {\em
  The logic of quantum mechanics}, Academic Press 1980.
\bibitem{CPavPaq} B.~Coecke, E.~O.~Paquette and D.~Pavlovic, Classical and quantum structuralism, in I.~Mackie and S.~Gay (eds), Semantic Techniques for Quantum Computation, pages 29--69, Cambridge University Press, 2009.
\bibitem{D'Ariano2006a} G.~M. d'Ariano, How to derive the Hilbert-space formulation of quantum mechanics from purely operational axioms, arXiv.org {\tt quant-ph/0603011} (2006).
\bibitem{DaviesLewis} E.~B.~Davies and J.~T.~Lewis, An operational approach to
quantum probability, Comm. Math. Phys. {\bf 17} 239-260 (1970).
%\bibitem{Dixon-Duncan}, Dixon-Duncan.. ?
\bibitem{Edwards} C. M. Edwards, The operational approach to
algebraic quantum theory I, Comm. Math. Phys. {\bf 16} (1970) 207-230.
\bibitem{FK} J. Faraut and A. Kor\'{a}nyi, {\em Analysis on Symmetric Cones}, Oxford,
1994.
\bibitem{GMCD10} D. Gross, M. M\"uller, R. Colbeck and O.~C.~O. Dahlsten, All reversible
dynamics in maximally nonlocal theories are trivial. Phys. Rev. Lett. {\bf 104} (2010), 080402.
%\bibitem{Gudder} S. Gudder,...
\bibitem{HO} H. Hanche-Olsen, On the structure and tensor products
of JC algebras, Can. J. Math. {\bf 35} (1983), 1059-1074.
\bibitem{Hanche-Olsen} H. Hanche-Olsen, JB-algebras with tensor
  products are $C^{\ast}$-algebras, \emph{Lecture Notes in
  Mathematics} (Springer) {\bf 1132} (1985) pp. 223-229.
\bibitem{Heunen} C. Heunen, \emph{Categorical quantum models and
  logics}, Doctoral thesis, Radboud University, Nijmegen.  Amsterdam,
  Pallas Publications, 2009.
\bibitem{JvNW} P. Jordan, J. von Neumann and E.~P. Wigner, On an algebraic
generalization of the quantum-mechanical formalism, \emph{Annals of Mathematics}
 {\bf 35} (1934), 29-64.
\bibitem{Klay} M. Kl\"{a}y, Einstein-Podolsky-Rosen experiments: the
  structure of the sample space, {\em Foundations of Physics Letters}
  {\bf 1} (1988), 205-244.
\bibitem{Koecher58}
M. Koecher, Die geod\"{a}tischen von positivit\"{a}tsbereichen,
Math. Annalen {\bf 135} (1958), 192-202.
\bibitem{Koecher62} M. Koecher, On Real Jordan algebras, Bull. Amer. Math. Soc. {\bf 68} (1962), 374-377.
\bibitem{KoecherNotes} M. Koecher, {\em The Minnesota Notes on Jordan
  Algebras and their Applications} edited and annotated by A. Krieg
  and S. Walcher, Lecture Notes in Mathematics {\bf 1710}, Springer Verlag,
 1999.
\bibitem{Keyl-Werner} M. Keyl and R. Werner, Channels and Maps, in
  D. Bruss and G. Leuchs, {\em Lectures on Quantum Information},
  pp. 73-86, Wiley, 2007
\bibitem{Kelly-Laplaza} G.~M. Kelly and M.~L. Laplaza, Coherence for
  compact closed categories, J Pure Appl Algebra (1980) 193-213.
\bibitem{KRF} M. Klay, C. H. Randall, and D. J. Foulis, Tensor
products and probability weights, Int. J. Theor. Phys. {\bf 26}
(1987) 199-219
\bibitem{Ludwig} G. Ludwig, Foundations of Quantum
  Mechanics I, Springer, 1983.
\bibitem{Mackey} G. Mackey, Mathematical Foundations of Quantum
Mechanics, Addison-Wesley, 1963
\bibitem{MacLane} S. MacLane, Categories for the Working
  Mathematician, Springer (1971).  Second edition (1997).
\bibitem{SelingerCPM} P. Selinger, Dagger compact closed categories
  and completely positive maps (extended abstract), in Proceedings of
QPL 05. Electronic Notes in Theoretical Computer Science {\bf 170}, 139--163 (2007).
\bibitem{Selinger} P. Selinger, Towards a semantics for higher-order quantum computation.  In \emph{Proceedings of the 2nd International Workshop on Quantum Programming Languages, Turku, Finland}.  TUCS General Publication No. 33, pp. 127-143, June 2004.
\bibitem{SelingerExample} P. Selinger, personal communication.
\bibitem{ShortBarrett09} A.~J. Short and J.~Barrett, Strong non-locality: A tradeoff between states and measurements. arXiv:quant-ph/0909.2601.
\bibitem{Vinberg} E. B. Vinberg, Homogeneous cones, Dokl. Acad. Nauk.
  SSSR {\bf 141} (1960) 270-273; English trans. Soviet
  Math. Dokl. {\bf 2}, 1416-1619 (1961).
\bibitem{vonNeumann} J. von Neumann, Mathematische Grundlagen der Quantenmechanik, Berlin, Springer 1932.  English translation: Mathematical Foundations of
Quantum Mechanics, Princeton, Princeton  University Press, 1955.
\bibitem{Wilce92} A. Wilce, Tensor products in generalized measure
  theory, Int. J. Theor. Phys. {\bf 31}, 1915 (1992) .
\bibitem{Wilce09} A. Wilce, Four and a half axioms for finite
  dimensional quantum mechanics, electronic preprint {\tt arxiv:0912.5530}.
\end{thebibliography}
\end{document}